\documentclass[11pt]{article}
\usepackage{amsmath,amssymb,amsthm}
\usepackage{url}

\newcommand{\2}{\vspace{0.2 cm}}
\newtheorem{theorem}{Theorem}[section]
\newtheorem{lemma}[theorem]{Lemma}

\newtheorem{proposition}[theorem]{Proposition}

\newtheorem{example}[theorem]{Example}

\newcommand\cA{\mathcal{A}}

\newcommand\co{\mathrm{co}}
\newcommand\sco{\mathcal{CO}}
\newcommand\cc{\mathrm{cc}}
\newcommand\scc{\mathcal{CC}}

\newcommand\cs{{$\mathcal{CS}$}}

\begin{document}
\date{\today}
\title{Algorithms for Generating Convex Sets in Acyclic Digraphs}
\author{P. Balister\thanks{Department of Mathematical Sciences,
University of Memphis, TN 38152-3240, USA. E-mail:
pbalistr@memphis.edu} \and S. Gerke\thanks{Department of Mathematics, Royal Holloway,
University of London, Egham, TW20 0EX, UK, E-mail:
stefanie.gerke@rhul.ac.uk} \and G. Gutin\thanks{Department of Computer Science, Royal Holloway,
University of London, Egham, TW20 0EX, UK, E-mail:
gutin@cs.rhul.ac.uk}\and A. Johnstone \thanks{Department of Computer Science, Royal Holloway,
University of London, Egham, TW20 0EX, UK, E-mail:
adrian@cs.rhul.ac.uk} \and J. Reddington\thanks{Department of Computer Science, Royal Holloway,
University of London, Egham, TW20 0EX, UK, E-mail:
joseph@cs.rhul.ac.uk} \and E. Scott\thanks{Department of Computer Science, Royal Holloway,
University of London, Egham, TW20 0EX, UK, E-mail:
eas@cs.rhul.ac.uk}\and A. Soleimanfallah\thanks{Department of Computer Science, Royal Holloway,
University of London, Egham, TW20 0EX, UK, E-mail:
arezou@cs.rhul.ac.uk} \and A. Yeo\thanks{Department of Computer Science, Royal Holloway,
University of London, Egham, TW20 0EX, UK, E-mail:
anders@cs.rhul.ac.uk}}

\maketitle

\begin{abstract}
A set $X$ of vertices of an acyclic digraph $D$ is convex if $X\neq
\emptyset$ and there is no directed path between vertices of $X$
which contains a vertex not in $X$. A set $X$ is connected if $X\neq
\emptyset$ and the underlying undirected graph of the subgraph of
$D$ induced by $X$ is connected. Connected convex sets and convex
sets of acyclic digraphs are of interest in the area of modern
embedded processor technology. We construct an algorithm $\cal A$
for enumeration of all connected convex sets of an acyclic digraph
$D$ of order $n$. The time complexity of $\cal A$ is $O(n\cdot
cc(D))$, where $cc(D)$ is the number of connected convex sets in
$D$. We also give an optimal algorithm for enumeration of all (not
just connected) convex sets of an acyclic digraph $D$ of order $n$.
In computational experiments we demonstrate that our algorithms
outperform the best algorithms in the literature.

Using the same approach as for $\cal A$, we design an algorithm for
generating all connected sets of a connected undirected graph $G$.
The complexity of the algorithm is $O(n\cdot c(G)),$ where $n$ is
the order of $G$ and $c(G)$ is the number of connected sets of $G.$
The previously reported algorithm for connected set enumeration is
of running time $O(mn\cdot c(G))$, where $m$ is the number of edges
in $G.$
\end{abstract}

\section{Introduction}

A set $X$ of vertices of an acyclic digraph $D$ is {\em convex} if
$X\neq \emptyset$ and there is no directed path between vertices of
$X$ which contains a vertex not in $X$. A set $X$ is {\em connected}
if $X\neq \emptyset$ and the underlying undirected graph of the
subgraph of $D$ induced by $X$ is connected. A set is {\em connected
convex} (a {\em cc-set}) if it is both connected and convex.

In Section \ref{ccsec}, we introduce and study an algorithm $\cal A$
for generating all connected convex sets of a connected acyclic
digraph $D$ of order $n$. The running time of $\cal A$ is $O(n\cdot
cc(D))$, where $cc(D)$ is the number of connected convex sets in
$D$. Thus, the algorithm is (almost) optimal with respect to its
time complexity. Interestingly, to generate only $k$ cc-sets using
$\cal A$ we need $O(n^3+kn)$ time. In Section \ref{expersec}, we
give experimental results demonstrating that the algorithm is
practical on reasonably large data dependency graphs for basic
blocks generated from target code produced by
Trimaran~\cite{trimaran} and SimpleScalar~\cite{simplescalar}. Our
experiments show that $\cal A$ is better than the state-of-the-art
algorithm of Chen, Maskell and Sun \cite{chen}. Moreover, unlike the
algorithm in \cite{chen}, our algorithm has a provable (almost)
optimal worst time complexity.

Although such algorithms are of less importance in our application
area because of wider scheduling issues, there also exist algorithms
that enumerate all of the convex sets of an acyclic graph. Until
recently the algorithm of choice for this problem was that of Atasu,
Pozzi and Ienne~\cite{atasu2003, atasu2006}, however the CMS
algorithm~\cite{chen} (run in general mode) outperforms the API
algorithm in most cases. In Section \ref{all}, we give a different
algorithm, for enumeration of all the convex sets of an acyclic
digraph, which significantly outperforms the CMS and API
algorithms and which has a (optimal) runtime performance of the
order of the sum of the sizes of the convex sets.

Avis and Fukuda \cite{avisDAM65} designed an algorithm for
generating all connected sets in a connected graph $G$ of order $n$
and size $m$ with time complexity $O(mn\cdot c(G))$ and space
complexity $O(n+m)$, where $c(G)$ is the number of connected sets in
$G$. Observe that when $G$ is bipartite there is an orientation $D$
of $G$ such that every connected set of $G$ corresponds to a cc-set
of $D$ and vice versa. To obtain $D$ orient every edge of $G$ from
$X$ to $Y$, where $X$ and $Y$ are the partition classes of $G$.

The algorithm of Avis and Fukuda is based on a so-called reverse
search. Applying the approach used to design the algorithm $\cal A$ to connected
set enumeration, in Section \ref{consec}, we describe an algorithm
$\cal C$ for generating all connected sets in a connected graph $G$
of order $n$ with much better time complexity, $O(n\cdot c(G))$.
This demonstrates that our approach can be applied with success to
various vertex set/subgraph enumeration problems. The space
complexity of our algorithm matches that of the algorithm of Avis
and Fukuda.

\subsection{Algorithms Applications}\label{aasec}

There is an immediate application for $\cal A$ in the field of
so-called {\em custom computing} in which central processor
architectures are parameterized for particular applications.

An embedded or {\em application specific} computing system only ever
executes a single application. Examples include automobile engine
management systems, satellite and aerospace control systems and the
signal processing parts of mobile cellular phones. Significant
improvements in the price-performance ratio of such systems can be
achieved if the instruction set of the application specific
processor is specifically tuned to the application.

This approach has become practical because many modern integrated
circuit implementations are based on Field Programmable Gate Arrays
(FPGA). An FPGA comprises an array of logic elements and a
programmable routing system, which allows detailed design of logic
interconnection to be performed directly by the customer, rather
than a complete (and very high cost) custom integrated circuit
having to be produced for each application. In extreme cases, the
internal logic of the FPGA can even be modified whilst in operation.

Suppliers of embedded processor architectures are now delivering
{\em extensible} versions of their general purpose processors.
Examples include the ARM OptimoDE~\cite{arm}, the MIPS Pro
Series~\cite{mips} and the Tensilica Xtensa~\cite{tens}. The
intention is that these architectures be implemented either as
traditional logic with an accompanying FPGA containing the hardware
for extension instructions, or be completely implemented within a
large FPGA. By this means, hardware development has achieved a new
level of flexibility, but sophisticated design tools are required to
exploit its potential.

The goal of such tools is the identification of time critical or
commonly occurring patterns of computation that could be directly
implemented in custom hardware, giving both faster execution and
reduced program size, because a sequence of base machine
instructions is being replaced by a single custom {\em extension}
instruction. For example, a program solving simultaneous linear
equations may find it useful to have a single instruction to perform
matrix inversion on a set of values held in registers.

The approach proceeds by first locating the {\em basic blocks} of
the program, regions of sequential computation with no control
transfers into them. For each basic block we construct a {\em data
dependency graph} (DDG) which contains vertices for each base
(unextended) instruction in the block, along with a vertex for each
initial input datum. Figure~\ref{ddg} shows an example of a DDG.
There is an arc to the vertex for the instruction $u$ from each
vertex whose instruction computes an input operand of $u$. DDG's are
acyclic because execution within a basic block is by definition
sequential.

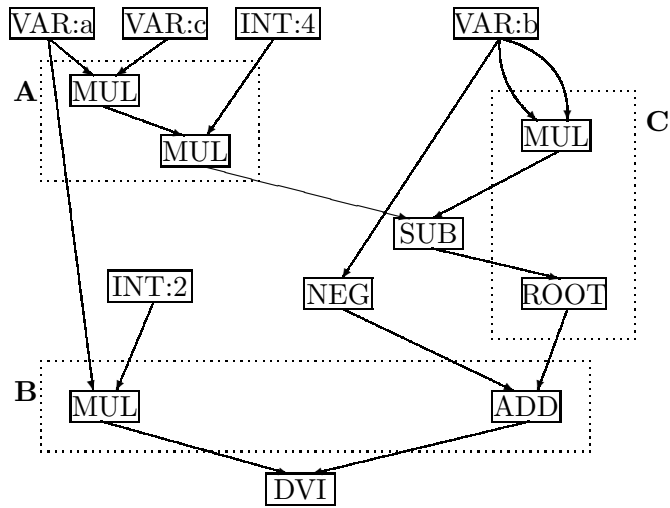
\begin{figure}
\begin{center}
\unitlength 1mm 
\linethickness{0.4pt}
\ifx\plotpoint\undefined\newsavebox{\plotpoint}\fi 
\begin{picture}(88,66)(0,0)
\put(2,62){\framebox(11,4)[cc]{VAR:a}}
\put(17,62){\framebox(11,4)[cc]{VAR:c}}
\put(32,62){\framebox(11,4)[cc]{INT:4}}
\put(61,62){\framebox(11,4)[cc]{VAR:b}}
\put(70,26){\framebox(11,4)[cc]{ROOT}}
\put(15,27){\framebox(11,4)[cc]{INT:2}}
\put(10,53){\framebox(9,4)[cc]{MUL}}
\put(22,45){\framebox(9,4)[cc]{MUL}}
\put(10,11){\framebox(9,4)[cc]{MUL}}
\put(53,34){\framebox(9,4)[cc]{SUB}}
\put(70,47){\framebox(9,4)[cc]{MUL}}
\put(66,11){\framebox(9,4)[cc]{ADD}}
\put(41,26){\framebox(9,4)[cc]{NEG}}
\put(36,0){\framebox(9,4)[cc]{DVI}}
\put(13,57){\vector(4,-3){0}}\multiput(7,62)(.04026846,-.03355705){149}{\line(1,0){.04026846}}
\put(16,57){\vector(-4,-3){0}}\multiput(23,62)(-.04697987,-.03355705){149}{\line(-1,0){.04697987}}
\put(25,49){\vector(3,-1){0}}\multiput(14,53)(.09243697,-.03361345){119}{\line(1,0){.09243697}}
\put(28,49){\vector(-2,-3){0}}\multiput(37,62)(-.033707865,-.048689139){267}{\line(0,-1){.048689139}}
\put(13,15){\vector(1,-4){0}}\multiput(7,62)(.03370787,-.26404494){178}{\line(0,-1){.26404494}}
\put(42,4){\vector(-4,-1){0}}\multiput(71,11)(-.13942308,-.03365385){208}{\line(-1,0){.13942308}}
\put(39,4){\vector(4,-1){0}}\multiput(14,11)(.12019231,-.03365385){208}{\line(1,0){.12019231}}
\put(76,51){\vector(0,-1){0}}\qbezier(67,62)(76.5,59.5)(76,51)
\put(72,51){\vector(1,-1){0}}\qbezier(67,62)(66.5,55.5)(72,51)
\put(58,38){\vector(-2,-1){0}}\multiput(75,47)(-.063670412,-.033707865){267}{\line(-1,0){.063670412}}
\put(27,45){\vector(4,-1){28}}
\put(75,30){\vector(4,-1){0}}\multiput(58,34)(.14285714,-.03361345){119}{\line(1,0){.14285714}}
\put(72,15){\vector(-1,-3){0}}\multiput(76,26)(-.03361345,-.09243697){119}{\line(0,-1){.09243697}}
\put(69,15){\vector(2,-1){0}}\multiput(46,26)(.070336391,-.033639144){327}{\line(1,0){.070336391}}
\put(46,30){\vector(-2,-3){0}}\multiput(67,62)(-.033707865,-.051364366){623}{\line(0,-1){.051364366}}
\put(16,15){\vector(-1,-2){0}}\multiput(21,27)(-.03355705,-.08053691){149}{\line(0,-1){.08053691}}
\put(66,22){\makebox(19,33)[cc]{}}
\multiput(65.93,21.93)(.95,0){21}{{\rule{.4pt}{.4pt}}}
\multiput(65.93,54.93)(.95,0){21}{{\rule{.4pt}{.4pt}}}
\multiput(65.93,54.93)(0,-.97059){35}{{\rule{.4pt}{.4pt}}}
\multiput(84.93,54.93)(0,-.97059){35}{{\rule{.4pt}{.4pt}}}
\put(6,7){\makebox(73,12)[cc]{}}
\multiput(5.93,6.93)(.98649,0){75}{{\rule{.4pt}{.4pt}}}
\multiput(5.93,18.93)(.98649,0){75}{{\rule{.4pt}{.4pt}}}
\multiput(5.93,18.93)(0,-.9231){14}{{\rule{.4pt}{.4pt}}}
\multiput(78.93,18.93)(0,-.9231){14}{{\rule{.4pt}{.4pt}}}
\put(6,43){\makebox(29,16)[cc]{}}
\multiput(5.93,42.93)(.9667,0){31}{{\rule{.4pt}{.4pt}}}
\multiput(5.93,58.93)(.9667,0){31}{{\rule{.4pt}{.4pt}}}
\multiput(5.93,58.93)(0,-.9412){18}{{\rule{.4pt}{.4pt}}}
\multiput(34.93,58.93)(0,-.9412){18}{{\rule{.4pt}{.4pt}}}
\put(4,55){\makebox(0,0)[cc]{{\bf A}}}
\put(88,51){\makebox(0,0)[cc]{{\bf C}}}
\put(4,15){\makebox(0,0)[cc]{{\bf B}}}
\end{picture}
\end{center}
\caption{Data dependency graph for $\frac{-b + \sqrt{b^2-4ac}}{2a}$}
\label{ddg}
\end{figure}

Extension instructions are combinations of base machine instructions
and are represented by sets of the DDG. In Figure~\ref{ddg},
sections A and B are convex sets that represent candidate extension
instructions. However, Section B is not connected.  If such a region
were implemented as a single extension instruction we should have
separate independent hardware units within the instruction. Although
this presents no special difficulties, and in Section~\ref{all} we
give an optimal algorithm for constructing all such sets, present
engineering practice is to restrict the search to connected convex
components on the grounds that unconnected convex components are
composed of connected ones, and that the system's code scheduler
will perform better if it is allowed to arrange the independent
computations in different ways at different points in the program.

Unlike connectivity however, convexity is not optional. An extension
instruction cannot perform computations that depend on instructions
external to the extension instruction. This means that there can be
no data flows out of and then back into the extension instruction:
the set corresponding to an extension instruction must be convex.
Thus section C in Figure~\ref{ddg} does not represent a candidate
extension instruction since it breaches the `no external computation
rule' because it is non-convex: there is a path {\em via} the {\tt
SUB} node that is not in the set.

Ideally we would like to fully consider all possible candidate
instructions and select the combination which results in the most
efficient implementation. In practice this is unlikely to be
feasible as, in worst case, the number of candidates will be
exponential in the number of original program instructions. However,
it is useful to have a process which can find all the potential
instructions, even if the set of instructions used for final
consideration has to be restricted. In this work we only deal with
generation of a set of possible candidate instructions. Interested
readers can refer to~\cite{atasu2006,yu}.

\subsection{ Related Theoretical Research}\label{rtrsec}

Many other algorithms for special vertex set/subgraph generation
have been studied in the literature. Kreher and Stinson
\cite{kreher} describe an algorithm for generating all cliques in a
graph $G$ of order $n$ with running time $O(n\cdot cl(G))$, where
$cl(G)$ is the number of cliques in $G.$

Several algorithms have been suggested for the generation of all spanning
trees in a connected graph $G$ of order $n$ and size $m$. Let $t$ be
the number of spanning trees in $G$. The first spanning trees
generating algorithms \cite{gabowSIAMJC7,mintyIEEETCT12,readN5} used
backtracking which is useful for enumerating various kinds of
subgraphs such as paths and cycles. Using the algorithms from
\cite{mintyIEEETCT12,readN5}, Gabow and Myers \cite{gabowSIAMJC7}
suggested an algorithm with time complexity $O(tn+n+m)$ and space
complexity $O(n+m)$. If we output all spanning trees by their edges,
this algorithm is optimal in terms of time and space complexities.
Later algorithms of a different type were developed; these algorithms
(see, e.g., \cite{kapoorSIAMJC24,shiouraJORSJ38,shiouraSIAMJC26})
find a new spanning tree by exchanging a pair of edges. As a result,
the algorithms of Kapoor and Ramesh \cite{kapoorSIAMJC24} and
Shioura and Tamura \cite{shiouraJORSJ38} require only $O(t+n+m)$
time and $O(nm)$ space. The algorithm of Shioura, Tamura and Uno
\cite{shiouraSIAMJC26} is of the same optimal running time, but also
of optimal space: $O(n+m).$

An out-tree is an orientation of a tree such that all vertices but
one are of in-degree 1. Kapoor, Kumar and Ramesh \cite{kapoorA27}
presented an algorithm for enumerating all spanning out-trees of a
digraph with $n$ vertices, $m$ arcs and $t$ spanning out-trees. The
algorithm takes $O(\log n )$ time per spanning tree; more precisely,
it runs in $O(t \log n +n^2\alpha(n,n)+nm)$, where $\alpha$ is the
Inverse Ackermann function. It first outputs a single spanning
out-tree and then a list of arc swaps; each spanning out-tree can be
generated from the first spanning out-tree by applying a prefix of
this sequence of arc swaps.

\section{Terminology, Notation and Preliminaries}

Let $D$ be a digraph. If $xy$ is an arc of $D$ ($xy\in A(D)$), we
say that $y$ is an {\em out-neighbor} of $x$ and $x$ is an {\em
in-neighbor} of $y$. The set of out-neighbors of $x$ is denoted by
$N^+_D(x)$ and the set of in-neighbors of $x$ is denoted by
$N^-_D(x)$. For a set $X$ of vertices of $D$, its {\em
out-neighborhood} (resp.\ {\em in-neighborhood}) is $N^+_D(X)=\bigcup_{x\in
X}N^+_D(x)\setminus X$ (resp.\ $N^-_D(X)=\bigcup_{x\in X}N^-_D(x)\setminus
X$). A digraph $D^{TC}$ is called the {\em transitive closure} of
$D$ if $V(D^{TC})=V(D)$ and a vertex $x$ is an in-neighbor of a
vertex $y$ in $D^{TC}$ if and only if there is a path from $x$ to
$y$ in $D.$

Let $S$ be a non-empty set of vertices of a digraph $D$. A directed
path $P$ of $D$ is an $S$-{\em path} if $P$ has at least three
vertices, its initial and terminal vertices are in $S$ and the rest
of the vertices are not in $S.$  For a digraph $D$, $\scc(D)$
($\sco(D)$) denotes the collection of cc-sets (convex sets) in $D;$
$\cc(D)=|\scc(D)|$ and $\co(D)=|\sco(D)|.$ An ordering
$v_1,v_2,\ldots, v_n$ of vertices of an acyclic digraph $D$ is
called {\em acyclic} if for every arc $v_i v_j$ of $D$ we have
$i<j$.

\begin{lemma} \label{lemma_transitive}
Let $D$ be a connected acyclic digraph and let $S$ be a vertex set
in $D$. Then $S$ is a cc-set in $D$ if and only if it is a cc-set in
$D^{TC}$.
\end{lemma}
\begin{proof} Let $S$ be a set of vertices of $D$. We will first prove that
there is an $S$-path in $D$ if and only if there is an $S$-path in
$D^{TC}.$ Since all arcs of $D$ are in $D^{TC}$, every $S$-path in
$D$ is an $S$-path in $D^{TC}.$ Let $Q=x_1x_2\ldots x_q$ be an
$S$-path in $D^{TC}$. Then there are paths $P_2,P_3,\ldots ,P_q$
such that $Q'=x_1P_2x_2P_3x_3\ldots x_{q-1}P_qx_q$ is a path in $D$
($Q'$ must be a path since $D$ is acyclic). Since $x_1$ and $x_q$
belong to $S$ and $x_2$ does not belong to $S$, there is a subpath
of $Q'$ which is an $S$-path.

If $S$ is connected in $D$ then it is clearly connected in $D^{TC}$,
which implies that if $S$ is a cc-set in $D$ then it is a cc-set in
$D^{TC}$. Now let $S$ be a cc-set in $D^{TC}$. Assume that $D[S]$ is
not connected and let $x$ and $y$ be vertices in different connected
components in $D[S]$, but which are connected by an arc in $D^{TC}$.
Without loss of generality $xy$ is the arc in $D^{TC}$ and $Q$ is a
path from $x$ to $y$ in $D$. However as $S$ is convex all vertices
in $Q$ also belong to $S$ and therefore $x$ and $y$ belong to the
same connected component in $D[S]$, a contradiction. \end{proof}

\vspace{3mm}

It is well-known  (see, e.g., the paper
\cite{fisher1971}  by Fisher and Meyer, or \cite{furmanSMD11} by
Furman) that the transitive closure problem and the matrix
multiplication problem are closely related: there exists an
$O(n^a)$-algorithm, with $a\ge 2$, to compute the transitive
closure of a digraph of order $n$ if and only if the product of
two boolean $n\times n$ matrices can be computed in $O(n^a)$ time. 
Coppersmith and Winograd \cite{coppersmith1987}
showed that there exists an $O(n^{2.376})$-algorithm for the
matrix multiplication. Thus, we have the following:

\begin{theorem}\label{remtr}
The transitive closure of a digraph of order $n$ can be found in
$O(n^{2.376})$ time. 
\end{theorem}

We will need the following two results proved in \cite{GuYe}.

\begin{theorem} \label{lower_bound}
For every connected acyclic  digraph $D$ of order $n$,  $cc(D) \geq
  n(n+1)/2$. If an acyclic digraph $D$ of order $n$ has a Hamiltonian path, then
  $cc(D)=n(n+1)/2$.
\end{theorem}

\begin{theorem}\label{upper_bound}
Let $f(n)=2^n+n+1-d_n$, where $d_n=2\cdot 2^{n/2}$ for every even
$n$ and $d_n=3\cdot 2^{(n-1)/2}$ for every odd $n$. For every
connected acyclic  digraph $D$ of order $n$,  $cc(D) \leq f(n).$ Let
$\vec{K}_{p,q}$ denote the digraph obtained from the complete
bipartite graph $K_{p,q}$ by orienting every edge from the partite
set of cardinality $p$ to the partite set of cardinality $q.$ We
have $\cc(\vec{K}_{a,n-a})=f(n)$ provided $|n-2a|\le 1$.
\end{theorem}

\section{Algorithm for Generating CC-Sets of an Acyclic
Digraph}\label{ccsec}

In this section $D$ denotes a connected acyclic digraph of order $n$
and size $m$. Now we describe the main algorithm of this paper; we
denote it by ${\cal A}$. The input of ${\cal A}$ is $D$ and ${\cal
A}$ outputs all cc-sets of $D$. The formal description of ${\cal A}$
is followed by an example and proofs of correctness of ${\cal A}$
and its complexity. Finally, we show that to produce $k$ cc-sets
$\cal A$ requires $O(n^{2.376}+kn)$ time. The algorithm works as follows. 
Given a digraph
$D$ on $n$ vertices, it considers an acyclic ordering $v_1,\ldots,v_n$ of the
transitive closure of $D$.
For each vertex $v_i$ we consider the sets $X=\{v_i\}$ and $Y=\{v_{i+1},\ldots ,
v_{n}\}$ and call the subroutine ${\cal B}(X,Y,D)$ which finds all cc-sets $S$ in
$D$ such
that $X \subseteq S \subseteq X \cup Y$.
At each step, if possible ${\cal B}(X,Y,D)$ removes an element $v$ from $Y$ and adds
it to $X$. If $X$ has out-neighbors we choose $v$ to be the
`largest' out-neighbor in the acyclic ordering(line 3), otherwise if $X$ has
in-neighbors
we choose $v$ to be the `smallest' in-neighbor (line 8).  Then we
find the other vertices required to maintain convexity (line 4 or
line 9). If there are no in- or out-neighbors we output $X$,
otherwise we find find all the cc-sets such that $X\subseteq
S\subseteq X\cup Y$ and $v\in S$ (line 12) and then all the cc-sets
such that $X\subseteq S\subseteq X\cup Y$ and $v\not\in S$ (line
13).

\begin{description}
  \item[Step 1:] Find the transitive closure of $D$ and
  set $D =D^{TC}.$
  \item[Step 2:] Find an acyclic ordering $v_1,v_2,\ldots,v_n$ of
  $D$.
  \item[Step 3:] For each $i=1,2,\ldots,n$ do the following.  Set $X:=\{v_i\}$,
$Y:=\{v_{i+1},v_{i+2}, \ldots, v_{n} \}$ and call ${\cal B}(X,Y,D)$.

\item[Step 4 subroutine ${\cal B}(X,Y,D)$:]\mbox{}\\
{ \obeylines\obeyspaces{
 \ 1.   set $A = N^+_{D^{TC}}(X)\cap Y$
 \ 2.  {\bf if} $A\not=\emptyset$ $\{$
 \ 3.        set $v = v_j$, where $j=$max$\{i:\ v_i\in A\}$
 \ 4.        set $R = \{v\}\cup (N^-_{D^{TC}}(v)\cap A)$  $\}$
 \ 5.  {\bf else} $\{$
 \ 6.        set $B = N^-_{D^{TC}}(X)\cap Y$
 \ 7.       {\bf if} $B\not=\emptyset$ $\{$
 \ 8.             set $v = v_k$, where $k=$min$\{i:\ v_i\in B\}$
 \ 9.             set $R = \{v\}\cup (N^+_{D^{TC}}(v)\cap B)$  $\}$ $\}$
 \ 10. {\bf if} $A=\emptyset$ and $B=\emptyset$ $\{$ output $X$ $\}$
 \ 11. {\bf else} $\{$
 \ 12.        ${\cal B}(X\cup R,\ Y\backslash R,\ D)$
 \ 13.        ${\cal B}(X,\ Y\backslash \{v\},\ D)$ $\}$ $\}$
}}
\end{description}
Before proving the correctness of $\cal A$, we consider an example.

\begin{example} Let $D$ be the graph on the left below
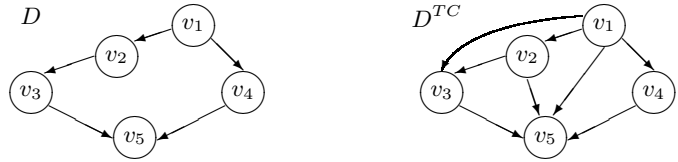
\begin{center}
\footnotesize
\input{ex1.pic}
\end{center}
In Step 1, we find $A(D^{TC})=A(D)\cup \{v_1v_3,v_2v_5,v_1v_5\}$
(above right). Observe that $v_1,v_2,v_3,v_4,v_5$ is an acyclic
ordering. We may assume that this is the ordering found in Step 2.

For $i=1$ in Step 3, we have $X=\{v_1\}$ and
$Y=\{v_2,v_3,v_4,v_5\}=N^+(X)$, and we call ${\cal B}(\{v_1\},
\{v_2,v_3,v_4,v_5\}, D)$. Then in Step 4, line 1, we compute
$A=\{v_2,v_3,v_4,v_5\}$ and then, lines 3 and 4, obtain $v=v_5$,
$N^-_{D^{TC}}(v)=\{v_1,v_2,v_3,v_4\}$ and $R=\{v_2,v_3,v_4,v_5\}$.
Then, at line 12, we make a recursive call to ${\cal
B}(V(D),\emptyset,D)$. In this call we have $A=B=\emptyset$ so, at
line 10, the set $V(D)=\{v_1,\ldots,v_5\}$ is output and the
recursive call returns, to line 13 of ${\cal B}(\{v_1\},
\{v_2,v_3,v_4,v_5\},D)$, where we make a call to ${\cal B}(\{v_1\},
\{v_2,v_3,v_4\},D)$. We are now effectively looking at the graph
$D_1$ below.
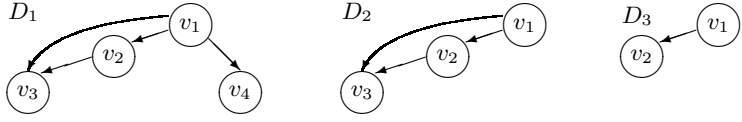
\begin{center}
\footnotesize
\input{ex2.pic}
\end{center}
In Step 4, lines 1-4, we compute $A=\{v_2,v_3,v_4\}$ and obtain
$v=v_4$, $N^-_{D^{TC}}(v)=\{v_1\}$ and $R=\{v_4\}$.  At lines 12 and
13 we make recursive calls to ${\cal
B}(\{v_1,v_4\},\{v_2,v_3\},D)$ and ${\cal B}(\{v_1\},\{v_2,v_3\},D)$
respectively.

In the call to ${\cal B}(\{v_1,v_4\},\{v_2,v_3\},D)$, lines 1-4, we
obtain $v=v_3$ and $R=\{v_2,v_3\}$. This in turn generates calls to
${\cal B}(\{v_1,v_4, v_2,v_3\},\emptyset,D)$, which just outputs
$\{v_1,v_2,v_3,v_4\}$ and returns, and ${\cal
B}(\{v_1,v_4\},\{v_2\},D)$. The latter call generates calls to
${\cal B}(\{v_1,v_4, v_2\},\emptyset,D)$ and ${\cal
B}(\{v_1,v_4\},\emptyset,D)$, which output $\{v_1,v_2,v_4\}$ and
$\{v_1,v_4\}$, respectively.

In the call to ${\cal B}(\{v_1\},\{v_2,v_3\},D)$, where we are
effectively looking at $D_2$ above, we obtain $v=v_3$ and
$R=\{v_2,v_3\}$. This in turn generates calls to ${\cal
B}(\{v_1,v_2,v_3\},\emptyset,D)$, which just outputs
$\{v_1,v_2,v_3\}$ and returns, and ${\cal B}(\{v_1\},\{v_2\},D)$
(graph $D_3$ above). The latter call generates calls to ${\cal
B}(\{v_1, v_2\},\emptyset,D)$ and ${\cal B}(\{v_1\},\emptyset,D)$,
which output $\{v_1,v_2\}$ and $\{v_1\}$, respectively. This
completes the case $i=1$ in Step 3, and all the cc-sets containing
$v_1$ have been output.

Now we perform Step 3 with $i=2$, effectively looking at the graph
$D_4$.
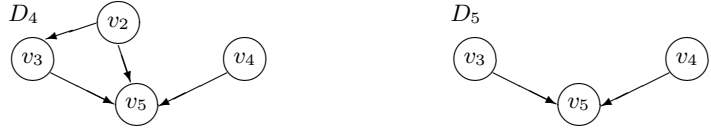
\begin{center}
\footnotesize
\input{ex3.pic}
\end{center}
The call to ${\cal B}(\{v_2\},\{v_3,v_4,v_5\},D)$ 
generates further recursive calls in the following order\\
\mbox{\qquad} ${\cal B}(\{v_2,v_5,v_3\},\{v_4\},D)$\\
\mbox{\qquad\qquad} ${\cal B}(\{v_2,v_5,v_3,v_4\},\emptyset,D)$, output
$\{v_2,v_3,v_4,v_5\}$\\
\mbox{\qquad\qquad} ${\cal B}(\{v_2,v_5,v_3\},\emptyset,D)$, output $\{v_2,v_3,v_5\}$\\
\mbox{\qquad} ${\cal B}(\{v_2\},\{v_3,v_4\},D)$\\
\mbox{\qquad\qquad} ${\cal B}(\{v_2,v_3\},\{v_4\},D)$, output $\{v_2,v_3\}$\\
\mbox{\qquad\qquad} ${\cal B}(\{v_2\},\{v_4\},D)$, output $\{v_2\}$.\\
Thus all the cc-sets containing $v_2$ but not $v_1$ are output.

Performing Step 3 again with $i=3$, effectively looking at the graph
$D_5$ above, the call to ${\cal B}(\{v_3\},\{v_4,v_5\},D)$,
generates the following recursive calls\\
\mbox{\qquad} ${\cal B}(\{v_3,v_5\},\{v_4\},D)$\\
\mbox{\qquad\qquad} ${\cal B}(\{v_3,v_5,v_4\},\emptyset,D)$, output $\{v_3,v_4,v_5\}$\\
\mbox{\qquad\qquad} ${\cal B}(\{v_3,v_5\},\emptyset,D)$, output $\{v_3,v_5\}$\\
\mbox{\qquad} ${\cal B}(\{v_3\},\{v_4\},D)$, output $\{v_3\}$\\
which ouput all the cc-sets containing $v_3$ but not $v_1$ or $v_2$.

For the case $i=4$ in Step 3 we get the following calls\\
\mbox{\qquad} ${\cal B}(\{v_4\},\{v_5\},D)$\\
\mbox{\qquad\qquad} ${\cal B}(\{v_4,v_5\},\emptyset,D)$, output $\{v_4,v_5\}$\\
\mbox{\qquad\qquad} ${\cal B}(\{v_4\},\emptyset,D)$, output $\{v_4\}$\\
and for $i=5$ we get\\
\mbox{\qquad} ${\cal B}(\{v_5\},\emptyset,D)$, output $\{v_5\}$\\
after which $\cal A$ terminates.
\end{example}

\2

\begin{lemma}\label{corlem}
Algorithm $\cal A$ correctly outputs all cc-sets of $D$.
\end{lemma}
\begin{proof}
Recall, the convex (connected) sets of $D$ are precisely the convex
(connected) sets of $D^{TC}$. We prove the result for $D^{TC}$.

Firstly we show that all the sets $X$ output by $\cal A$ are in
$\mathcal{CC}(D^{TC})$. We will show that within $\cal A$, for any call
${\cal B}(X,Y,D)$ we have that $X\cap Y=\emptyset$, $X\cup Y$ is
convex and $X$ is a cc-set. This is clearly sufficient as $X$ is the only set output.

These properties hold for Step 3 when ${\cal
B}(\{v_i\},\{v_{i+1},\ldots,v_n\},D)$ is called as we have chosen an acyclic
ordering of the vertices.  Thus we assume that the properties hold for the sets
$X$, $Y$ and consider the pairs of sets $X\cup R$, $Y\backslash R$ and $X$,
$Y\backslash \{v\}$ constructed in ${\cal B}(X,Y,D)$. In both cases
clearly the intersections are empty, and since $R\subseteq
N^+_{D^{TC}}(X)\cup N^-_{D^{TC}}(X)$, $X\cup R$ is connected.

Now we will prove that $X\cup R$ is convex. Suppose that there is a path $u,y,w$
where $u,w\in X\cup R$. Note that if there exists an $(X\cup R)$-path then by
transitivity of $D^{TC}$ there exists an $(X\cup R)$-path of length two. By
convexity of $X\cup Y$ we have $y\in X\cup Y$. Also, $y\neq v$ as we have chosen $v$
to be either the maximal element of $N^+_{D^{TC}}(X)$ or the minimal element of
$N^-_{D^{TC}}(X)$, and $D^{TC}$ is transitive and thus the presence of the arcs $uy$
and $yw$ implies the presence of the arc $uw.$ Assume that $A\neq \emptyset$. Then
$R\subseteq N^+_{D^{TC}}(X)$. Since $u \in X \cup N^+_{D^{TC}}(X)$ and the arc $uy$
exists, the transitivity of $D^{TC}$ implies that $y\in N^+_{D^{TC}}(X)$. Since $X$
is convex it follows that not both vertices  $u,w$ can be in $X$  and that there is
no arc from $N^+_{D^{TC}}(X)$ to $X$. Thus $w\not\in X$ and so 
$w\in R\subseteq N^-_{D^{TC}}(X)$. By the transitivity of $D^{TC}$ and the fact that
$yw$ exists and that $w\in N^-_{D^{TC}}(X)$ we have $y\in N^-_{D^{TC}}(X)$ and thus
$y\in R$. Similarly if  $A=\emptyset$ then $R \subseteq N^-_{D^{TC}}(X)$ and by the
transitivity of $D^{TC}$ and since $w\in X\cup  N^-_{D^{TC}}(X)$  we have  $u\in R$
and thus $y\in N^-_{D^{TC}}(X)\cap N^+{D^{TC}}(X)$.

Secondly we show that if $X\not=\emptyset$ is cc then $X$ is output
by $\cal A$.
If $S$ is a cc-set and $j=\min\{i:\ v_i\in S\}$ then $\{v_j\}
\subseteq S\subseteq \{v_j,v_{j+1}, \ldots, v_{n} \}$. Thus it is
sufficient to show that if $S$ is cc and $X \subseteq S \subseteq X
\cup Y$ then ${\cal B}(X,Y,D)$ outputs $S$. We prove this by
induction on $|Y|$.

If $(N^+_{D^{TC}}(X)\cap Y)
=\emptyset=(N^-_{D^{TC}}(X)\cap Y)$ then, since $S$ is connected,
$S=X$ and ${\cal B}(X,Y,D)$ outputs $X$ at line 10. This proves the
result for $|Y|=0$, and for $|Y|\geq 1$ we may assume that $v\in
(N^+_{D^{TC}}(X)\cup Y \cup N^-_{D^{TC}}(X))$.

If $v\not\in S$ then we have $X\subseteq
S\subseteq(X\cup(Y\backslash\{v\}))$ and $|Y\backslash\{v\}|<|Y|$, so
by induction the call to ${\cal B}(X,Y\backslash\{v\},D)$ at line 13
outputs $S$. If $r\in (R\backslash\{v\})$, we have arcs $rv$ and
$xr$, for some $x\in X\subseteq S$. Thus, if $v\in S$, by convexity
of $S$ we have $R\subseteq S$. Then, since $|Y\backslash R|<|Y|$,
the call to ${\cal B}(X\cup R,Y\backslash R,D)$ at line 12 outputs
$S$.
\end{proof}

\begin{lemma}\label{runtlem}
The running time of $\cal A$ is $O(n\cdot cc(D)).$
\end{lemma}
\begin{proof} Note that by Theorem \ref{lower_bound} and the fact that $D$ is
connected we have $n \times cc(D) \geq n^2(n+1)/2$. Therefore the
transitive closure of $D$ can be found in $O(n \cdot cc(D))$ time,
by Theorem \ref{remtr}. It is well-known that an acyclic ordering can
be found in time $O(n+m)$, see, e.g., \cite{bang2000}, and clearly
the sets $N^+_{D^{TC}}(v)$ and $N^-_{D^{TC}}(v)$ can be computed at
the start of the algorithm in $O(n)$ time, for each $v\in V(D)$.

We will now show that ${\cal B}(X,Y,D)$ runs in time $O(|Y| \cdot
cc'(X,Y)+K_{X,Y})$, where $cc'(X,Y)$ is the number of cc-sets $S$
such that $X \subseteq S \subseteq X \cup Y$ and $K_{X,Y}$ is the
sum of the sizes of the sets $S$. Note that ${\cal B}$ returns at
line 10 or makes two recursive calls to ${\cal B}$ (lines 12,13). If
${\cal B}$ returns at line 10 then we call this a {\em leaf} call
otherwise the function call is an {\em internal} call. All function
calls can be viewed as nodes of a binary tree (every node is a leaf
or has two children) whose leaves and internal nodes correspond to
calls to $\cal B$. It is easy to see, by induction, that the number
of internal nodes equals the number of leaves minus one. It is easy
to see, by induction on the depth of the call tree, that ${\cal B}$
outputs each set $S$ only once (${\cal B}(X\cup R,Y\backslash R,D)$
and ${\cal B}(X,Y\backslash\{v\},D)$ output those that contain $v$
and do not contain $v$, respectively). Thus we have $cc'(X,Y)$ leaf
calls and $cc'(X,Y)-1$ internal calls.

We assume that the set implementation allows us to find the size of
a set and the largest and smallest elements of the set in unit time.
Then the time taken by a call  ${\cal B}(X,Y,D)$ depends on the time
taken to calculate the sets $A$, $B$ and $R$. Since $A,B\subseteq
Y$, the time to compute $R$ is at most $O(|Y|)$. If we implement
${\cal B}(X,Y,D)\cap Y$ so that $N^+_{D^{TC}}(X)\cap Y$ and $N^-_{D^{TC}}(X)$
are passed in as parameters then the time taken to calculate $A$ and
$B$ is at most $O(|Y|)$. By definition of $R$ we have that  $N^+_{D^{TC}}(X
\cup R) = N^+_{D^{TC}}(X) - R$ and $N^-_{D^{TC}}(X \cup R) = N^-_{D^{TC}}(X) \cup
N^-_{D^{TC}}(v) - R -
X$ provided $A \not= \emptyset$, and $N^-_{D^{TC}}(X \cup R) = N^-_{D^{TC}}(X) - R$
and $N^+_{D^{TC}}(X \cup R) = N^+_{D^{TC}}(X) \cup N^+_{D^{TC}}(v) - R - X$ provided
$A=\emptyset$ (and $B \not= \emptyset$). Since $R\subseteq Y$, these
sets can be computed in $O(|Y|)$ time.

If ${\cal B}(X,Y,D)$ calls ${\cal B}(X',Y',D)$ then $|Y'|<|Y|$ thus
a call to ${\cal B}$ at an internal node takes at most  $O(|Y|)$ time,
and a call at a leaf node takes at most $O(|Y|+|X|)$ time, giving
the desired total time bound of $O(|Y| \cdot cc'(X,Y)+K_{X,Y})$.

We let $K_i$ denote the sum of the sizes of all the cc-sets $S$ such
that $v_i\in S\subseteq \{v_{i+1},\ldots,v_n\}$, and observe that
$K_1+\ldots +K_n\leq n\cdot cc(D)$.

Finally, by Step 3, we conclude that the total running time is
$$O\left(\sum_{i=1}^{n}cc'(\{v_i\},\{v_{i+1},v_{i+2},\ldots ,v_n\})\cdot
(n-i)+K_i\right)=O(cc(D)\cdot n).$$\end{proof}

\begin{theorem}\label{maint}
Algorithm $\cal A$ is correct and its time and space complexities
are $O(n\cdot cc(D))$ and $O(n^2)$, respectively.
\end{theorem}
\begin{proof} The correctness and time complexity follows from the two lemmas
above. The space complexity is dominated by the space complexity of
Step 1, $O(n^2).$\end{proof}

\2

Since $cc(D)$ may well be exponential, we may wish to generate only
a restricted number $k$ of cc-sets. Theorem \ref{kccsets} can be
viewed as a result in fixed-parameter algorithmics \cite{downey1999}
with $k$ being a parameter.

\begin{theorem}\label{kccsets}
To output $k$ cc-sets the algorithm $\cal A$ requires $O(n^{2.376}+kn)$
time.
\end{theorem}
\begin{proof} We may assume that $k$ is at most the number
of cc-sets containing vertex $v_1$ since otherwise the proof is
analogous.

We consider the binary tree $T$ introduced in the proof of Lemma
\ref{runtlem} and prove our claim by induction on $k$. It takes
$O(n^{2.376})$ time to perform Steps 1,2 and 3. It takes $O(n)$ internal
nodes of $T$ to reach the first leaf of $T$ and, thus, for $k=1$ we
obtain $O(n^{2.376}+n)$ time. Assume that $k\ge 2$. Let $x$ be the first
leaf of $T$ reached by $\cal A$, let $y$ be the parent of $x$ on
$T$, let $z$ be another child of $y$ on $T$ and let $u$ be the
parent of $y$. Observe that after deleting the nodes $x$ and $y$ and
adding an edge between $u$ and $z$, we obtain a new binary tree
$T'$. By induction hypothesis, to reach the first $k-1$ leaves in
$T'$, we need $O(n^{2.376}+(k-1)n)$ time. To reach the first $k$ leaves in
$T$, we need to reach $x$ and the first $k-1$ leaves in $T'$. Thus,
we need to add to $O(n^{2.376}+(k-1)n)$ the time required to visit $x$ and
$y$ only, which is $O(n).$ Thus, we have proved the desired bound
$O(n^{2.376}+kn)$.
\end{proof}

\section{Generating Convex Sets in Acyclic Digraphs}\label{all}

It is not hard to modify $\cal A$ such that the new algorithm will
generate all convex sets of an acyclic digraph $D$ in time $O(n\cdot
co(D))$, where $co(D)$ is the number of convex sets in $D$. However,
a faster algorithm is possible and we present one in this section.

To obtain all convex sets of $D$ (and $\emptyset$, which is not
convex by definition), we call the following recursive procedure
with the original digraph $D$ and with $F=\emptyset$. This call
yields an algorithm whose properties are studied below.

A vertex $x$ is a {\em source} ({\em sink}) if it has no in-neighbors (out-neighbors).
In general, the procedure {\cs}\  takes as input an acyclic digraph
$D=(V,A)$ and a set $F\subseteq V$ and outputs all convex sets of $D$
which contain $F$.  The procedure {\cs}\ outputs $V$ and then
considers all sources and sinks of the graph that are not in $F$. 
For each such source or sink $s$, we call {\cs}$(D-s,F)$ and then
add $s$ to $F$. Thus, for each sink or source $s\in V\setminus F$
we consider all sets that contain $s$ and all sets that do not
contain $s$.

{\obeylines\obeyspaces{

     {\cs}($D=(V,A),F$)
     1.  {\bf output} $V$
     2.  {\bf for all} $s\in V\setminus F$ with $|N^+(s)|=0$ or $|N^-(s)|=0$ {\bf do} \{
     3.       {\bf for all} vertices $v$ find $N^+_{D-s}(v)$ and $N^-_{D-s}(v)$
     4.        call {\cs}$(D-s,F)$; set $F:=F\cup \{s\}$
     5.       {\bf for all} vertices $v$ find $N^+_{D}(v)$ and $N^-_{D}(v)$            \}

}}

\subsection{Correctness of the procedure}

Proposition \ref{prop:convex} and Theorem \ref{prop:unique} imply
that the procedure {\cs}\ is correct. We first show that all sets
generated in line 1 are, in fact, convex sets. To this end, we use
the following lemma.

\begin{lemma}\label{lem:delete}
Let $D$ be an acyclic graph, let $X$ be a convex set of $D$, and let
$s\in X$ be a source or sink of $D[X]$. Then $X\setminus \{s\}$ is a
convex set of $D$.
\end{lemma}
\begin{proof}
Suppose that  $X\setminus \{s\}$ is not convex in $D$. Then there
exist two vertices $u,v\in X\setminus \{s\}$ and a directed path $P$
from $u$ to $v$ which contains a vertex not in $X\setminus \{s\}$.
Since $X$ is convex, $P$ only uses vertices of $X$ and in particular $s\in P$. 
Thus, there is a subpath $u'sv'$ of $P$ with $u',v'\in X$.
But since $s$ is a source or a sink in $D[X]$ such a subpath cannot
exist, a contradiction.
\end{proof}

Now we can prove the following proposition.

\begin{proposition}\label{prop:convex}
Let $D=(V,A)$ be an acyclic digraph and let $F\subseteq V$. Then
every set output by {\cs}($D,F$) is convex.
\end{proposition}
\begin{proof}
We prove the result by induction on the number of vertices of the
outputted set. The entire vertex set $V$ is convex and is outputted
by the procedure. Now assume all sets of size $n-i\geq 2$ that are
outputted by the procedure are convex. We will show that all sets of
size $n-i-1$ that are outputted are also convex. When a set $C$ is
outputted the procedure  {\cs}$(D[C],F')$ was called for some set
$F'\subseteq V$. The only way {\cs}$(D[C],F')$ can be invoked is
that there exist a set $C'\subset V$ and a source or sink $c$ of
$D[C']$ with $C=C'\setminus \{c\}$. Moreover $C'$ will be outputted
by the procedure and, thus, by our assumption is convex. The result
now follows from Lemma~\ref{lem:delete}.
\end{proof}

\begin{theorem}\label{prop:unique}
Let $D=(V,A)$ be an acyclic digraph and let $F\subseteq V$. Then
every convex set of $D$ containing $F$ is outputted exactly once by
{\cs}$(D,F)$.
\end{theorem}
\begin{proof}
Let $C$ be a convex set of $D$ containing $F$. We first claim that
there exist vertices $c_1,c_2,\ldots,c_t \in V$ with
$V=\{c_1,c_2,\ldots,c_t\}\cup C$ and $c_i$ is a source or sink of
$D[C\cup \{c_i,c_{i+1},\ldots,c_t\}]$ for all $i\in \{1,2,\ldots
,t\}$. To prove the claim we will show that for every convex set $H$
with
 $C\subset H\subseteq V$, there exists a source or sink $s\in H\setminus C$ of
the digraph $D[H]$. This will prove our claim as by
Lemma~\ref{lem:delete} $H\setminus \{s\}$ is a convex set of $D$ and we
can repeatedly apply the claim.

If there exists no arc from a vertex of $C$ to  a vertex of
$D[H\setminus C]$ then any source of $H\setminus C$ is a source of
$D[H]$. Note that $D[H\setminus C]$ is an acyclic digraph and, thus,
has at least one source (and sink). Thus we may assume that there is
an arc from a vertex $u$ of $C$ to a vertex $v$ of $H\setminus C$.
Consider a longest path $v=v_1v_2\ldots v_r$ in $D[H\setminus C]$
leaving $v$. Observe that $v_r$ is a sink of $D[H\setminus C]$ and,
moreover, there is no arc from $v_r$ to any vertex of $C$ since
otherwise there would be a directed path from $u\in C$ to a vertex
in $C$ containing vertices in $H\setminus C$ which is impossible as
$C$ is convex. Hence $v_r$ is a sink of $D[H]$ and the claim is
shown.

Next note that a sink or source remains a sink or source when
vertices are deleted. Thus when {\cs}$(D,F)$ is executed and  a
source or sink $s$ is considered, then we distinguish the cases when
$s=c_i$ for some $i\in \{1,2,\ldots ,t\}$ or when this is not the
case. If $s=c_i$  and we currently consider the digraph $D'$ and the
fixed set $F'$, then we follow the execution path calling
{\cs}$(D'-s,F')$. Otherwise we follow the execution path that adds
$s$ to the fixed set. When the last $c_i$ is deleted, we call
{\cs}$(D[C],F'')$  for some $F''$ and the set $C$ is outputed. It
remains to show that there is a unique execution path yielding $C$.
To see this, note that when we consider a source or sink $s$ then
either it is deleted of moved to the fixed set $F$. Thus every
vertex is considered at most once and then deleted or fixed.
Therefore each time we consider a source or sink there is a unique
decision that finally yields $C$.
\end{proof}

\subsection{Running time of \cs}

We assume that the input acyclic digraph $D=(V,A)$ is given by
the two adjacency lists for each vertex, and the number of in-neighbors and
out-neighbors is stored for each vertex. One can obtain this
information at the beginning in $O(n+m)$ time, where $n$ ($m$) is
the number of vertices (arcs) of the input connected acyclic digraph $D$. Observe
that we output the vertex set of $D$ as one convex set. Thus, it
suffices to show that the running time of {\cs}$(D,F)$ without the
recursive calls is $O(|V|)$. This will yield the running time
$O(\sum_{C\in \sco(D)}|C|)$ of {\cs}\ by
Theorem~\ref{prop:unique}.

Since we have stored the number of in-neighbors and out-neighbors
for every vertex $v\in V$, we can determine \emph{all} sources and
sinks in $O(|V|)$ time. For the recursive calls of {\cs}\ we delete
one vertex and have to update the number of in- respectively
out-neighbors of all neighbors of the deleted vertex $s$. The vertex
$s$ has at most $|V|-1$ neighbors and we can charge the cost of the
updating information to the call of {\cs}$(D-s,F)$. Moreover we store the neighbours
of $s$ so that we can reintroduce them after the call of  {\cs}$(D-s,F)$. Moving the
sinks
and sources to $F$ needs constant time for each source or sink and
thus we obtain $O(|V|)$ time in total.

In summary we initially need $O(n+m)$ time, and then each call of
{\cs}$(D,F)$ is charged with $O(|V|)$ before it is called and then
additionally with $O(|V|)$ time during its execution. Since we
output a convex set of size $O(|V|)$, the total running time is
$O(n+m)+O(\sum_{C\in \sco(D)}|C|)$. Since
 $\sum_{C\in \sco(D)}|C|=\Omega(n^2)$ by Theorem \ref{lower_bound}, the running
time of {\cs}\  is $O(\sum_{C\in \sco(D)}|C|)$.

\section{Implementation and Experimental Results}\label{expersec}

In order to test our algorithms $\cal A$ and {\cs}  for
practicality we have implemented and run them on several instances
of DDG's of basic blocks. We have compared our algorithm with the
state-of-the-art algorithm of Chen, Maskell and Sun \cite{chen}
(the CMS algorithm) using their own implementation, but with the
code for I/O constraint checking removed so as to ensure that their
algorithm was not disadvantaged. For completeness we have also
compared {\cs}\ to Atasu, Pozzi and Ienne's
algorithm~\cite{atasu2006} (the API06 algorithm). All the algorithms
were coded in C++ and all experiments were carried out on a 2 x Dual
Core AMD Opteron 265 1.8GHz processor with 4 Gb RAM, running SUSE
Linux 10.2 (64 bit).

Our first set of tests is based on C and C++ programs taken from the
benchmark suites of MiBench~\cite{MIBENCH} and
Trimaran~\cite{trimaran}. We compiled these benchmarks for both the
Trimaran (A,B,C,D,E) and SimpleScalar~\cite{simplescalar} (F,G,H,I)
architectures. From here we examined the control-flow graph for each
program to select a basic block within a critical loop of the
program (often this block had been unrolled to some degree to
increase the potential for efficiency improvements).

We considered basic blocks, ranging from  20 to 45 lines of low
level, intermediate, code, for which we generated the DDGs. We then
selected, from these DDGs, the non-trivial connected components on
which to run our algorithms.

We give some preference to benchmarks which suite the intended
application of the research taking our test cases from security
applications including benchmarks for the Advanced Encryption
Standard (B,C) and safety-critical software (A, E). We also include
a basic example from the Trimaran benchmark suite: Hyper (D), an
algorithm that performs quick sort (F), part of a jpeg algorithm
(G), and an example from the fft benchmark in mibench containing C
source code for performing Discrete Fast Fourier Transforms (H). The
final example is taken from the standard blowfish benchmark, an
encryption algorithm.

The results we have obtained are given in Table~\ref{ri:basic}.  In
the following tables NV denotes the number of vertices, NS denotes
the number of generated sets, NA number of arcs, CT denotes clock
time in $10^{-3}$ CPU seconds, and for the benchmark data ID
identifies the benchmark.

\begin{table}
\begin{center}
\begin{tabular}{|c|c|c|c||c||c|
} \hline ID&NV&NA&NS&CMS (CT)&$\cA$ (CT)\cr\hline\hline A &35 &38
&139,190 & 170& 96\cr\hline B &42 & 45& 4,484,110& 5,546 &
3,246\cr\hline C &26&28&5,891&6&4\cr\hline D
&39&94&3,968,036&4,346&2,710\cr\hline E
&45&44&1,466,961&1,750&1,156\cr\hline F &24&22&46,694&60&30\cr\hline
G &20&19&397&0&0\cr\hline H &20&21&1,916&0&0\cr\hline I
&43&47&10,329,762&13,146&7,210\cr\hline
\end{tabular}
\end{center}
\caption{cc-sets for benchmark programs} \label{ri:basic}
\end{table}
For examples G and H both algorithms ran in almost 0 time. For the
other examples, the above results demonstrate that our algorithm
$\cA$ outperforms the CMS algorithm.

We also consider examples with worst-case numbers of cc-sets. Let,
as in Theorem \ref{upper_bound}, $\vec{K}_{p,q}$ denote the digraph
obtained from the complete bipartite graph $K_{p,q}$ by orienting
every edge from the partite set of cardinality $p$ to the partite
set of cardinality $q.$ By Theorem \ref{upper_bound} the digraphs
$\vec{K}_{a,n-a}$ with $|n-2a|\le 1$ have the maximum possible
number of cc-sets. Our experimental results for digraphs
$\vec{K}_{a,n-a}$ with $|n-2a|\le 1$ are given in
Table~\ref{ri:max}. Again we see that $\cA$ outperforms the CMS
algorithm.

\begin{table}
\begin{center}
\begin{tabular}{|c|c|c||c||c|}
\hline NV &NA&NS&CMS (CT)&$\cA$ (CT)\cr\hline\hline 15& 56&  32,400&
30&16\cr\hline 16& 64&  65,041&     56&23\cr\hline 17& 72&  130,322&
114&60\cr\hline 18& 81&  261,139&    240&113\cr\hline 19& 90&
522,722&    540&253\cr\hline 20& 100& 1,046,549&  1,080&513\cr\hline
21& 110& 2,094,102&  2,166&1,048\cr\hline 22& 121& 4,190,231&
4,086&2,156\cr\hline
\end{tabular}
\end{center}
\caption{cc-sets for graphs with maximum number of cc-sets}
\label{ri:max}
\end{table}

We have compared algorithm {\cs}\ with both CMS running in
`unconnected' mode and with API06. The examples used are the same as
in Table 1, however we do not give results for examples B, D, E and
I as these graphs produce an extremely large number of convex sets
and as a result, do not terminate in reasonable time. The results
are shown in Table~\ref{ri:basicB}. We can see that although CMS
generally out-performs API06, there are two cases where API06 is
marginally better. However, {\cs}\ is consistantly three to five
times faster than either of the other algorithms.

\begin{table}
\begin{center}
\begin{tabular}{|c|c|c|c||c||c||c|
} \hline ID&NV&NA&NS&API06&CMS (CT)&{\cs}\ (CT)\cr\hline\hline A &35
&38 &1,123,851 &2,560 &1,390& 270\cr\hline C
&26&28&120,411&250&120&40\cr\hline F
&24&22&3,782,820&3,250&3,630&860\cr\hline G
&20&19&122,111&70&120&30\cr\hline H
&20&21&55,083&110&110&20\cr\hline
\end{tabular}
\end{center}
\caption{All convex sets for benchmark programs} \label{ri:basicB}
\end{table}

For interest we have also compared API06, CMS and {\cs}\ on the
digraphs that have maximal numbers of cc-sets. The results are shown
in Table~\ref{ri:maxB}. Again, while CMS and API06 are roughly
comparable, {\cs}\ is a least twice as fast as both of them.

\begin{table}
\begin{center}
\begin{tabular}{|c|c|c||c||c||c|}
\hline NV &NA&NS&API06&CMS (CT)&{\cs}\ (CT)\cr\hline\hline 15& 56&
32,768&     40&40&10\cr\hline 16& 64&  65,536&     70&70&30\cr\hline
17& 72&  131,072&    140&130&60\cr\hline 18& 81&  261,144&
320&320&130\cr\hline 19& 90&  524,288&    720&700&320\cr\hline 20&
100& 1,046,575&  1,590&1,500&710\cr\hline 21& 110& 2,097,152&
3,320&3,010&1,500\cr\hline 22& 121& 4,194,304&
7,140&6,310&3,120\cr\hline
\end{tabular}
\end{center}
\caption{All convex sets for graphs with maximum number of cc-sets}
\label{ri:maxB}
\end{table}

\section{Connected Sets Generation Algorithm}\label{consec}

Let $G$ be a connected (undirected) graph with vertex set
$V(G)=\{v_1,v_2,\ldots,v_n\}$ and let $G$ have $m$ edges. For a
vertex $x\in V(G)$ and a set $X\subseteq V(G)$, let $N(x)=\{z\in
V(G):\ xz\in E(G)\}$ and $N(X)=\bigcup_{x\in X}N(x)\setminus X.$ The
following is an algorithm, $\cal C$, for generating all connected
sets of $G$.

\begin{description}
\item[Step 1:]
For each $i=1,2,\ldots,n$ do the following.  Set $X:=\{v_i\}$ and
$Y:=\{v_{i+1},v_{i+2}, \ldots, v_{n} \}$. Initiate the set $N_X$ as
$N_X:=N(X) \cap Y$.

\item[Step 2 (subroutine ${\cal D}$):]
{\em Comment: $\cal D$ finds all connected sets $Q$ in $D$ such that
$X \subseteq Q \subseteq X \cup Y$}.

\begin{description}
\item[(2a):] If $N_X=\emptyset$ then return the connected set $X$ (and stop).

\item[(2b):] If $N_X \not=\emptyset$, then let $v \in N_X$ be arbitrary.

\item[(2c):]  {\em Comment: In this step we will find all connected sets $S$
such that $X\cup\{v\} \subseteq S \subseteq (X \cup Y)$}.\\
Set
$N_{X,0}:=N_X$, $X_0:=X$ and $Y_0:=Y$. Remove $v$ from $Y$ and
$N_X$, and add it to $X$. For every $u \in Y\setminus N_X$ check
whether $u$ has an edge to $v$ and if it does then add it to $N_X$.

\2 Make a recursive call to subroutine ${\cal D}$. {\em Comment: we
consider the new $X$ and $Y$}.

\2 Change $N_X$, $X$, and $Y$ back to their original state by
setting $N_X:=N_{X,0}$, $X:=X_0$, and $Y:=Y_0$.

\item[(2d):] {\em Comment: In this step we will find all connected
sets $S$  such that $X \subseteq S \subseteq (X \cup Y)$ and $v
\not\in S$}. Remove $v$ from $Y$ and remove $v$ from $N_X$.

\2

Make a recursive call to subroutine ${\cal D}$.

\2

Change $Y$ back to its original state by adding $v$ back to $Y$.
Also change $N_X$ back to its original state by adding $v$ to it.
\end{description}
\end{description}

Similarly to Theorem \ref{maint}, one can prove the following:

\begin{theorem}
Let $c(G)$ be the number of connected sets of a connected graph $G$.
Algorithm $\cal C$ is correct and its time and space complexities
are $O(n\cdot c(G))$ and $O(n+m)$, respectively.
\end{theorem}

\section{Discussions and Open Problems}\label{discsec}

Our computational experiments show that $\cal A$ performs well and
is of definite practical interest. We have tried various heuristic
approaches to speed up the algorithm in practice, but all approaches
were beneficial for some instances and inferior to the original
algorithm for some other instances. Moreover, no approach could
significantly change the running time. The algorithm was developed
independently from the CMS algorithm. However, the two algorithms
are closely related, and work continues to isolate the
implementation effects that give the performance differences.

\2 \2

\noindent{\bf Acknowledgements.}  We are grateful to the authors of
\cite{chen} for helpful discussions and for giving us access to
their code allowing us to benchmark our algorithm against theirs.
Research of Gregory Gutin and Anders Yeo was supported in part by an
EPSRC grant. Research of Gutin was also supported in part by the IST
Programme of the European Community, under the PASCAL Network of
Excellence, IST-2002-506778.

\end{document}